\def\year{2022}\relax
\documentclass[letterpaper]{article} 
\usepackage{aaai22}  
\usepackage{times}  
\usepackage{helvet}  
\usepackage{courier}  
\usepackage[hyphens]{url}  
\usepackage{graphicx} 
\usepackage{natbib}  
\usepackage{caption} 
\DeclareCaptionStyle{ruled}{labelfont=normalfont,labelsep=colon,strut=off} 
\usepackage{amsmath}
\usepackage{amsthm}
\usepackage{amssymb,amsfonts,textcomp, mathtools}
\usepackage{todonotes}

\newcommand{\np}{{\mathrm{NP}}}

\newcommand{\cO}{\mathcal{O}}
\newcommand{\cE}{\mathcal{E}}
\newcommand{\cU}{\mathcal{U}}
\newcommand{\ie}{i.e.,\xspace}
\newcommand{\eg}{e.g.,\xspace}

\DeclareMathOperator{\rank}{rank}
\DeclareMathOperator{\rev}{rev}
\newtheorem{definition}{Definition}
\newtheorem{proposition}{Proposition}

\usepackage[ruled,vlined,linesnumbered,algo2e]{algorithm2e}

\usepackage{cleveref}
\usepackage{thm-restate}

\DeclareMathOperator{\alg}{ALG}
\DeclareMathOperator{\opt}{OPT}

\DeclareMathOperator*{\argmin}{arg\,min}
\newenvironment{proofsketch}{%
  \proof}{\endproof}
\usepackage{algorithm}
\usepackage{algorithmic}

\usepackage{newfloat}
\usepackage{listings}
\floatstyle{ruled}
\newfloat{listing}{tb}{lst}{}
\floatname{listing}{Listing}
\nocopyright
\title{Online Elicitation of Necessarily Optimal Matchings}
\author{Jannik Peters}
\affiliations{
    Research Group Efficient Algorithms\\
    TU Berlin, Germany\\
    jannik.peters@tu-berlin.de
}

\begin{document}
\maketitle
\begin{abstract}
    In this paper, we study the problem of eliciting preferences of agents in the house allocation model. For this we build on a recent model of \citet{HMSS21}[AAAI'21] and focus on the task of eliciting preferences to find matchings which are necessarily optimal, \ie optimal under all possible completions of the elicited preferences.  In particular, we follow the approach of \citet{HMSS21} and investigate the elicitation of necessarily Pareto optimal (NPO) and necessarily rank-maximal (NRM) matchings. Most importantly, we answer their open question and give an online algorithm for eliciting an NRM matching in the next-best query model which is $\frac{3}{2}$-competitive, \ie it takes at most $\frac{3}{2}$ as many queries as an optimal algorithm. Besides this, we extend this field of research by introducing two new natural models of elicitation and by studying both the complexity of determining whether a necessarily optimal matching exists in them, and by giving online algorithms for these models.
\end{abstract}
\section{Introduction}
One of the key settings in the area of matching under preferences is the so-called house allocation or assignment problem. In this problem we are given two sets, a set of agents $A$ and a set of houses $H$ with agents having preferences over houses. This simple setting has found multiple real life applications, for instance in the allocation of people to jobs \cite{HyZe79}, papers to reviewers \cite{GKK10a}, or students to student dorms \cite{ChSo02}.
Over the years, various solution concepts have been designed for the house allocation problem, for instance \emph{Pareto optimality} \cite{AbSo98a} \cite{ACMM04}, \emph{popularity} \cite{AIKM07} or \emph{rank-maximality} \cite{IKMMP06}.

However, most of the work on house allocation problems assumes the preferences of the agents to be given in their entirety, while in many real-world applications only partial preferences might be known and eliciting complete rankings from agents might be costly. 

As an expository (non-serious) example (based on a real life story), imagine a group of AI researchers meeting in their office kitchen to celebrate the acceptance of multiple papers. For this occasion, the researchers decide to eat some ice pops. However, after opening the freezer, they notice that only one ice pop of each kind is left, causing discussion on how to fairly divide the ice. Quickly, the group agrees that a rank-maximal allocation would be the fairest they could currently think of. Now there is just one problem left, due to time constraints and hunger, the researchers do not want to all give their whole ranking to each other. Instead, they agree that they should start off with naming the ice pop they like the most. But how should they continue after this, and who should be asked for their second favorite ice?

To deal with this problem \citet*{HMSS21} initiated the study of finding matchings that are necessarily optimal by eliciting partial preferences from the agents. In their model \citet{HMSS21} use so-called top-$k$ preferences in which each agent has only elicited a prefix of their true preferences. To obtain these preferences, they introduce the \emph{next-best} query model, in which agents can be asked to reveal the top house they have not revealed yet. The goal in this setting is to ask as few queries as possible in order to find a matching that is \emph{necessarily optimal}, i.e., optimal under every possible linear extension of the top-$k$ preferences. The performance of such an elicitation algorithm is then measured in terms of the so-called competitive ratio, i.e., the ratio between the number of queries of the algorithm and the number of queries of an optimal algorithm with knowledge of the complete preferences. As their main results, \citet{HMSS21} gave an $\mathcal{O}(\sqrt{n})$-competitive elicitation algorithm for finding a necessarily Pareto optimal matching,
    showed that no elicitation algorithm for finding a necessarily Pareto optimal matching can be $o(\sqrt{n})$-competitive,
    and proved that no elicitation algorithm for finding a necessarily rank-maximal matching can be $\frac{4}{3}-\varepsilon$-competitive for any $\varepsilon > 0$.
Further, they conjectured that an online algorithm with a constant competitive ratio for eliciting necessarily rank-maximal matchings is possible, and left this as their most important open question.
\subsection{Our results}

We contribute to this line of research in the following way.  First, we confirm the conjecture of \citet{HMSS21} and show that an online algorithm with a $\frac{3}{2}$-competitive ratio for eliciting a necessarily rank-maximal matching does exist in the next-best query model. Further, we show that this algorithm is optimal and no online-algorithm can have a competitive ratio better than $\frac{3}{2}$. Besides this, our many focus lies on the \emph{hybrid-query} model in which agents can be asked to either elicit a house of a given rank or to return the rank of a given house. In this model, we give an online algorithm with a constant competitive ratio for eliciting a necessarily rank-maximal matching, as well as, for any $\varepsilon > 0$,  an $\mathcal{O}(n^{\frac{1}{3}+\varepsilon})$-competitive algorithm for eliciting a necessarily Pareto optimal matching which almost meets the lower bound of $\Omega(n^\frac{1}{3})$. To add on to this, we also give a polynomial time algorithm for determining whether an NRM matching exists and show that the same problem becomes $\np$-complete for Pareto optimal matchings.

Finally, we also introduce the \emph{set-compare} model in which agents can be asked to give their top-choice element out of a set. Here, we show that this model is already powerful enough to obtain a $1$-competitive elicitation algorithm for Pareto optimal matchings. We show that determining whether an NPO matching under this preference model exists, is $\np$-complete as well, and that the $\frac{3}{2}$ lower bound obtained for the next-best model is also valid for the set-compare model. For a brief overview over the competitiveness bounds derived in this paper, we refer the reader to \Cref{ta:results}.
\begin{table}[tb]
	\begin{center}
			\resizebox{\columnwidth}{!}{ 
		\begin{tabular}{l c c c c c c c c} 
			
			& \multicolumn{2}{c}{next-best} & \multicolumn{2}{c}{set-compare} & \multicolumn{2}{c}{hybrid-query}  \\ \hline
			&LB & UB &LB & UB &LB & UB \\ 
			Pareto optimal & $\Omega(\sqrt{n})^\dag$ & $\mathcal O(\sqrt{n})^\dag$ & 1 & 1 & $\Omega(n^\frac{1}{3})$ & $\mathcal{O}(n^{\frac{1}{3} + \varepsilon})$ \\ \\
			Rank-maximal & $\frac{3}{2}$ & $\frac{3}{2}$ & $\frac{3}{2}$ & $\mathcal{O}(n)$ & $\frac{3}{2}$ & $6$ \\ 
			
		\end{tabular}
		}	
		
	\end{center}
	
	\caption{Overview over the lower (LB) and upper bounds (UB) on the possible competitiveness of online algorithms derived for eliciting Pareto optimal and rank-maximal matchings in the different query models. Results marked with $\dag$ were shown by \citet{HMSS21}.\label{ta:results}}
\end{table}%
\subsection{Related Work}
The house allocation or assignment problem is one of the key matching settings in both computer science and economics. Besides the aforementioned classical works of for instance \cite{BoMo01a, HyZe79, ShSc74a} recent papers on house allocation include work on envy-free house allocation \cite{GSV19, BCGHLMW19} on diversity constrains \cite{BCHSZ18}, incorporating cardinal queries in ordinal preferences \cite{MML21} or closely related to us on Pareto optimal house allocation under probabilistic uncertainty \cite*{ABDR19}.

Rank-maximal matchings were first introduced by \citet{irving2003greedy} and were subsequently studied and characterized by \citet{IKMMP06}. Following these two initial papers, several works studied algorithmic aspects of various variants of the problem \cite{KavSha06, Mich07, Paluch13,GNN19, NNP19}. Besides this \citet{BMW21} studied rank-maximality (and popularity) in a variant of the house allocation problem, where not only the allocation of the houses, but also the selection of the allocated houses, \ie which houses are matched and which are unmatched, matter. Very recently \citet{AziSun21} used rank-maximality and algorithmic techniques of \citet{IKMMP06} for the school choice problem with diversity constraints. 

Besides the aforementioned works by \citet{HMSS21, ABDR19} uncertainty in matching markets has been incorporated in several papers in the literature on two-sided matchings for instance by \citet{RCIL13, LMPS14}. Besides this, \citet{DruBou14} studied preference elicitation for the stable matching problem or \cite{GSSS17,MaiVaz18,CSS19} who studied stable matchings under various aspects of robustness, \eg stable under probabilistic perturbations or after a certain number of swaps in the input rankings. Finally, very closely related to our work is also the study of possible and necessary winners in computational social choice, where given partial preferences of voters, a candidate winning every election or some election is required. See \cite{Lang20} for a recent survey on this topic.

\section{Preliminaries}
For $a,b \in \mathbb{N}$ let $[a,b] = \{a, a+1,  \dots, b\}$ and $[a] = [1,a]$.

Throughout the paper, we let $A = \{a_1, \dots, a_n\}$ denote our set of \emph{agents} and $H = \{h_1, \dots, h_n\}$ our set of \emph{houses}. A \emph{matching} in our setting is simply a subset $M \subseteq A \times H$ such that no agent and no house appear in more than one pair. If $(a_i, h_j) \in M$ for some agent $a_i \in A$ and house $h_j \in H$ we say that $a_i$ is matched to $h_j$. 

Further, we assume that each agent $a_i \in A$ has a strict \emph{preference list} $\succ_i$ over all houses in $H$. If $h_j \succ_i h_k$ for two houses $h_j$ and $h_k$ we say that $a_i$ \emph{prefers} $h_j$ to $h_k$. When $h_j$ appears in the $k$th place in the preference list of $a_i$ we say that the \emph{rank} of $h_j$ in the preference list of $a_i$ is $k$ and write $\rank(a_i,h_j) = k$. For a given subset $H' \subseteq H$ and agent $a_i$, we call $\max_i(H')$ the maximum element of $H'$ with regard to $\succ_i$, \ie the house in $H'$ which $a_i$ likes the most. We refer to the collection of preference lists as a \emph{preference profile} $\succ$.

We are now ready to define the two problems we investigate in our paper. 
\paragraph{Pareto optimality.} We begin with the classical notion of Pareto optimal matchings \cite{AbSo98a}. Given a matching $M$, we say that another matching $M'$ \emph{dominates} $M$ if 
\begin{itemize}
    \item for every agent $a_i \in A$ it holds that $M'(a_i) \succ_i M(a_i)$ or $M'(a_i) = M(a_i)$ ;
    \item for at least one agent $a_i \in A$ it holds that $M'(a_i) \succ_i M(a_i)$.
\end{itemize} Now, a matching $M$ is \emph{Pareto optimal} if there is no matching $M'$ which dominates $M$.
\paragraph{Rank-maximality.} As our second optimality concept we consider rank-maximality. For a given matching $M$ let $r^M_l \coloneqq \lvert \{a_i \in A \mid \rank(a_i, M(a_i)) = l\} \rvert$ for any $l \in [n]$. Now a matching is \emph{rank-maximal} if and only if there is no other matching $M'$ and $l \in [n]$ such that $r^M_k = r^{M'}_k$ for all $k \in [l-1]$ and $r^M_l < r^{M'}_l$, \ie the vector $(r^M_1, \dots, r^M_n)$ is lexicographically maximal among all matchings. If such a matching $M'$ were to exist, we also say that $M'$ \emph{rank-dominates} $M$. Thus, a matching $M$ is rank-maximal if it first maximizes the number of agents matched to their first choice, subject to that maximizes the number of agents matched to their second choice and so on.
\subsubsection{Elicitation Protocols.}
Now, we turn to the three different elicitation protocols we study in our work. For the definition of the models, we assume that we are given a fixed instance of the house allocation problem with preference profile $\succ$.
\begin{itemize}
    \item First, we investigate the \emph{next-best} query model as defined by \citet{HMSS21}. In this model, we are only allowed to ask one type of query. Namely, we can query a single agent, who will return the house they rank the highest, which has not been revealed yet, \ie if this is the $k$th query asked to the agent, the query returns the house ranked $k$th by the agent in $\succ$. For any agent $a \in A$ we denote such a query as $\mathcal{Q}(a)$ and we refer to the set of agents revealed to by $A$ as $\rev(a)$.
    \item Next we study the \emph{hybrid-query} model. Here we can ask two types of queries. Firstly, a \emph{rank query} $\mathcal{Q}(a_i, k)$ for an agent $a_i \in A$ and $k \le n$ returns the house $h_j \in H$ with $\rank(a_i, h_j) = k$  and secondly a \emph{house query} $\mathcal{Q}(a_i, h_j$) for an $a_i \in A$ and a house $h_j \in H$ returns $\rank(a_i, h_j)$. Similarly to the next-best query model, for any $a\in A$ we refer to $\rev(a)$ as the set of houses $h \in H$ for which we know $\rank(a,h)$.
    \item As our third model, we study a less restricted version of the next-best query model, which we call the \emph{set-compare} query model. Here a query $\mathcal{Q}(a_i, H')$ consists of an agent $a_i \in A$ and a subset of houses $H' \subseteq H$ and returns $\max_{i}(H')$, \ie the house $a_i$ likes best in $H'$. This model is inspired by recent works of learning rankings in the area of machine learning \cite{CLM18, RLS19, SahGop19, SahGop20}.
    
\end{itemize}
For any of the three aforementioned query models, let $\mathcal{Q}_1, \dots, \mathcal{Q}_k$ be a sequence of queries with answers $\alpha_1, \dots, \alpha_k$ (we also refer to this as partial preferences throughout the paper). We call a preference profile $\succ$ consistent with these queries if the output of these queries on $\succ$ would be $\alpha_1, \dots, \alpha_k$ as well. A matching $M$ is now \emph{necessarily Pareto optimal(NPO)} (\emph{necessarily rank-maximal(NRM))} for a given sequence of queries if $M$ is Pareto optimal (rank-maximal) for all preference profiles consistent with these queries. The goal is now to design an online-algorithm which can ask queries according to one of the three aforementioned models and outputs a matching that is either necessarily Pareto optimal or necessarily rank-maximal according to the queries the algorithm asked. We assume that the online-algorithm only has access to the agents, houses and the answers to the queries, but not to the underlying preferences themselves.

In order to compare the performance of these algorithms, we measure their competitive ratio in comparison to an optimal algorithm which also knows the underlying preference profile. For any instance, such an optimal algorithm asks the minimum number of queries, after which a necessarily optimal matching with regard to these queries, asked by the optimal algorithm, can be given. We call an online algorithm $\alpha$-competitive if for any preference profile $\succ$ the online algorithm asks at most $\alpha \cdot\opt_\succ$ queries, where $\opt_\succ$ is the number of queries of the optimal algorithm on this instance.

Finally, we note that partial preferences in the next-best query model can be equivalently expressed by an incomplete preference profile $\succ'$ (with induced rank function $\rank')$, in the hybrid-query model, by an incomplete rank function $\rank'$ (with induced partial preference profile $\succ'$) in which each agent only lists ranks for a subset of houses, and in the set-compare model, by having a partial order $\succ'_i$ for each agent $a_i$.
\section{Pareto-optimal Matchings}
We start off with Pareto-optimal matchings. Here, \citet{HMSS21} managed to give an asymptotically tight $\mathcal{O}(\sqrt{n})$-competitive     elicitation algorithm for the next-best query model. As our main results, we first give a $1$-competitive algorithm for eliciting NPO matchings in the set-compare model, followed by a classification of NPO matchings in the hybrid-query model together with an $\mathcal{O}(n^{\frac{1}{3}+\varepsilon})$-competitive elicitation algorithm for any $\varepsilon > 0$. Before we turn to our elicitation algorithms, we quickly recap the famous serial dictatorship mechanism \cite{AbSo98a} and its relation to Pareto optimal matchings. 
\begin{definition}[Serial Dictatorship Mechanism]
The serial dictatorship mechanism takes as input a permutation $\sigma$ of the agents together with a preference profile $\succ$ and returns a matching $SD_\succ(\sigma)$ which iteratively matches agent $\sigma(i)$ to their most preferred house in $\succ$ not matched to by an agent in $\sigma(1), \dots, \sigma(i-1)$.
\end{definition}
As shown by \citet{AbSo98a} the serial dictatorship mechanism is already enough to classify all Pareto optimal matchings.
\begin{restatable}[\citet{AbSo98a}]{thm}{poclassify}
Given a preference profile $\succ$ a matching $M$ is Pareto optimal if and only if there is a permutation of the agents $\sigma$ such that $M = SD_\succ(\sigma)$.
\end{restatable}
This immediately brings us to the set-compare query model where we can show that this model is already sufficient to simulate the serial dictatorship mechanism, which allows us to construct a $1$-competitive algorithm.
\begin{restatable}{thm}{setcompareeli}
There exists a $1$-competitive algorithm in the set-compare model for computing a necessarily Pareto-optimal matching.
\end{restatable}
\begin{proof}
Our algorithm is a simple adaption of the serial dictatorship mechanism to the set-compare model. It works iteratively by constructing a matching $M$. In iteration $i$ let $H_i$ be the set of houses already matched by $M$ in previous iterations. Then in iteration $i$ we query $h \coloneqq \mathcal{Q}(a_i, H \setminus H_i)$ and add $(a_i, h)$ to $M$. Note that we do not need to query in iteration $n$ since only one agent/house pair is left. It is easy to see that for all possible preference extensions, in iteration $i$ agent $a_i$ is matched to the currently unmatched house they prefer the most. Therefore, this algorithm simulates the serial dictatorship mechanism and thus produces a (necessarily) Pareto optimal matching. Furthermore, the algorithm only uses $n-1$ queries and is therefore $1$-competitive, since at most $1$ agent can be left unqueried by the optimal algorithm.
\end{proof}
We further note that this proof can also be generalized to the setting where the set $H$ in each query can contain at most $k$ houses, \eg if $k = 2$ this would mean that only pair-wise comparisons could be asked to the agents. For a proof sketch, we refer to the full version.
For the hybrid-query model, we start off with the complexity of determining whether an NPO matching exists. While it is still polynomial time checkable if a given matching is NPO, we also show that it is $\np$-complete to determine the existence of an NPO matching.
\begin{restatable}{thm}{algnpomixed}
Given a matching $M$ and partial preferences $\rank'$ in the hybrid-query model, it can be determined in polynomial time whether $M$ is necessarily Pareto optimal.
\label{thrm:alg_npo_mixed}
\end{restatable}
\begin{proof}
This follows fairly simply by adapting the algorithm of \citet{HMSS21} for determining whether a matching $M$ is NPO in the next-best model. We create an auxiliary directed graph $G = (A, E)$ in which we add an arc from agent $a_i$  to agent $a_j$ if it is possible for $a_i$ to prefer $M(a_j)$ to $M(a_i)$. Then a cycle in $G$ implies that a matching dominating $M$ exists in a preference extension of $\rank'$. To be more precise, we add an edge from $a_i$ to $a_j$ if $M(a_i),M(a_j) \in \rev(i)$ and $\rank'
    (a_i, M(a_j)) < \rank'(a_i, M(a_i))$; or if $M(a_i) \in \rev(i), M(a_j) \notin \rev(i)$ and there is a rank $k < \rank(a_i, M(a_i))$ with no revealed house for $a_i$; or if $M(a_i) \notin \rev(i), M(a_j) \in \rev(i)$ and there is a rank $k > \rank(a_i, M(a_j))$ with no revealed house for $a_i$;
    or if $M(a_i), M(a_j) \notin \rev(i)$.
It is easy to see that there is a preference profile consistent with the partial preferences in which $a_i$ prefers $M(a_j)$ to $M(a_i)$ if and only if there is an edge from $a_i$ to $a_j$. Thus, a cycle in $G$ would indeed imply that we could extend the preferences in such a way, that we could construct a matching $M'$ dominating $M$, by swapping the houses along the cycle. On the other hand, if a matching $M'$ dominates $M$ in some preference extension, there has to be a cycle of agents $a_1, \dots, a_k, a_{k+1} = a_1$ with agent $a_i$ preferring $M(a_{i+1})$ to $M(a_i)$ for this preference extension and thus also forming a cycle in $G$. 
\end{proof}
This algorithm also translates into an algorithm for determining whether a matching $M$ is necessarily Pareto optimal in the set-compare model, again by adding edges from one agent to another, if there is any extension where one agent could prefer the house of the other agent.
\begin{restatable}{cor}{algnposc}
Given a matching $M$ and partial preferences $\succ'$ in the set-compare model it can be determined in polynomial time whether $M$ is necessarily Pareto optimal.
\label{thrm:alg_npo_sc}
\end{restatable}
Using the simple algorithm in \Cref{thrm:alg_npo_mixed} we can also give a succinct classification of necessarily Pareto optimal matchings in the hybrid-query model using the Serial Dictatorship mechanism. 
\begin{restatable}{lem}{nposerial}
Given partial preferences $\rank'$ in the hybrid-query model a matching $M$ is necessarily Pareto optimal if and only if there is a permutation $\sigma$ of $A$ such that for all possible preference extensions $\succ$ of $\rank'$ it holds that $M = SD_\succ(\sigma)$ .
\label{lem:npo_serial}
\end{restatable}
\begin{proof}
Let $M$ be an NPO matching and $G$ the graph constructed in Theorem \ref{thrm:alg_npo_mixed} for $M$. Since $M$ is NPO the graph $G$ is acyclic. Therefore, there exists a topological ordering of $G$. Let $\sigma$ be a reversed topological ordering of $G$. Then for every $i \in [n]$ and every possible preference extension (and thus also in $\succ$), the agent $\sigma(i)$ could only possibly prefer the houses already matched to the agents $\sigma(1), \dots, \sigma(i-1)$ if they were assigned their partner in $M$. Therefore, $SD_\succ(\sigma)$ would set $SD_\succ(\sigma)(\sigma(i)) = M(\sigma(i))$. Thus, by induction, we get that $SD_\succ(\sigma) = M$.
\end{proof}
This however does not translate to an algorithm for finding an NPO matching or determining that one exists. To show the $\np$-completeness of this problem we reduce from the $\np$-complete \textsc{(2,2)-e3-sat} problem, \cite{BKS04}. In an instance of the \textsc{(2,2)-e3-sat} problem we are given a set of variables $X$ and a set of clauses $\mathcal C$ over $X$ with each clause in $\mathcal{C}$ having length exactly $3$ such that each variable in $X$ appears exactly twice in negated form and twice in positive form in $\mathcal{C}$. 

\begin{restatable}{thm}{npcmqnpo}
Given partial preferences $\rank'$ in the hybrid-query model it is $\np$-complete to determine whether an NPO matching exists.
\end{restatable}
The proof of this theorem and all further missing proofs are in the appendix.
Using this result, we can also show that the same problem is $\np$-complete in the set-compare model. For this we simply show how to, given an instance in the hybrid-query model, construct an instance in the set-compare model, such that a matching is NPO in the former model if and only if is also NPO in the latter.
\begin{restatable}{thm}{scnponpc}
Given partial preferences $\succ'$ in the set-compare model, it is $\np$-complete to determine whether a necessarily Pareto-optimal matching exists.
\label{thrm:alg_setcomp_npc_po}
\end{restatable}
Even though it is inherently hard to find an NPO matching given partial preferences in the hybrid query model, we can still give an elicitation algorithm improving upon the competitive ratio for the next-best query model. 
Before proving this, we give a useful lower bound on the number of queries asked to an agent by using the serial dictatorship characterization of NPO matchings.
\begin{restatable}{lem}{lowerboundpo}
Let $\rank'$ be partial preferences and $M$ be an NPO matching in the hybrid-query model. Then there exists a permutation of agents $\sigma$ such that for all $i \in [n]$ agent $\sigma(i)$ has revealed at least $\min( \rank'(\sigma(i), M(\sigma(i))), n-i)$ of their preference list if $M(\sigma(i)) \in \rev(\sigma(i))$. If $M(\sigma(i)) \notin \rev(\sigma(i))$ the agent must have revealed at least $n-i$ houses.
\label{lem:lower_bound_po}
\end{restatable}
\begin{proof}
By Lemma \ref{lem:npo_serial} we know that there has to be a permutation of agents $\sigma$ such that $M = SD_\succ(\sigma)$ for all possible preference extensions $\succ$ of $\rank'$. Then since for every preference extension $SD_\succ(\sigma)$ matched $\sigma(i)$ to $M(\sigma(i))$ we know that all houses matched to $\sigma(i+1), \dots, \sigma(n)$ must be ranked lower than $M(\sigma(i))$ in all preference extensions. Thus, we either know the preferences of all houses matched to $\sigma(i+1), \dots, \sigma(n)$ and have thus revealed at least $n-i$ of the preference list of $\sigma(i)$ or there is at least one house matched to $\sigma(i+1), \dots, \sigma(n)$ we do not know the preference of. 
Then we must have queried the preferences of all houses ranked at least as high as $M(\sigma(i))$ thus requiring us to reveal at least $\rank'(\sigma(i), M(\sigma(i)))$ houses in the preference list of $\sigma(i)$.
\end{proof}
\begin{algorithm}[tb]
\caption{Elicitation algorithm for Pareto optimal matchings in the hybrid query model}
\label{alg:po_hybrid}
\textbf{Input}: Set of agents $A$, set of houses $H$, parameter $c_0 > \frac{1}{3}$.\\
\textbf{Output}: A necessarily Pareto optimal matching.
\begin{algorithmic}[1] 
\STATE set $E \gets \emptyset,M \gets \emptyset, j \gets 0$
\FOR{$i = 1, \dots, n$}
\FORALL{$a \in A$}
\STATE $E \gets E \cup \{\{a, \mathcal{Q}(a,i)\}\}$
\ENDFOR
\STATE $M \gets$ maximum size Pareto-optimal matching in $(A \cup H, E)$
\IF{$i = \lceil n^{c_j}\rceil$}
\IF{$\lvert M \rvert \ge n - \left\lceil n^{(c_j+1)/2}\right\rceil$}
\STATE break for loop 
\ELSE
\STATE $j \gets j + 1, c_j \gets (3c_{j-1} + 1)/2 + c_0 -1$
\ENDIF
\ENDIF
\ENDFOR
\STATE $H'\subseteq H \gets$ subset matched by $M$
\STATE $A'\subseteq A\gets$ subset matched by $M$
\FORALL{$a \in A \setminus A'$}
\STATE $h \gets \argmin_{h \in H \setminus H'} \mathcal{Q}(a,h)$
\STATE $M\gets M \cup \{\{a,h\}\}$ 
\STATE $H' \gets H' \cup \{h\}$
\ENDFOR
\end{algorithmic}
\end{algorithm}
The crucial idea behind our \Cref{alg:po_hybrid} for eliciting an NPO matching is now as follows. Assume that we want to construct an algorithm with competitive ratio $n^c$ for some constant $c > 0$ and that we know that at least $k$ agents must be matched to a house of at least rank $r$ with $k \ge r$, for instance by knowing that the maximum matching in the top-$r$ preferences has size $n-k$. Then by \Cref{lem:lower_bound_po} we know that at least $\frac{k}{2}$ agents must be asked at least $\min(r, \frac{k}{2})$ queries, since at least $\frac{k}{2}$ of the agents matched to a house with a rank of at least $r$ must be listed between $r$ and $n-\frac{k}{2}$ by $\sigma$. Thus, we know that any optimal algorithm must ask at least $\Omega(kr)$ queries which allows our online algorithm to ask $\mathcal{O}(n^c kr)$ queries. The trick behind \Cref{alg:po_hybrid} is now to choose these values of $c,k,$ and $r$ appropriately. To give some further idea behind the value of $\frac{1}{3}$ we show that $c_0 > \frac{1}{3}$ implies that the $(c_j)$ series in \Cref{alg:po_hybrid} is increasing. 
\begin{restatable}{lem}{increasing}
The series $(c_j)_{j \in \mathbb{N}}$ with $c_{j+1} =(3c_{j} + 1)/2 + c_0 -1$ is increasing if $c_0 > \frac{1}{3}$.
\label{lem:increasing}
\end{restatable}
Using this Lemma, we can now turn to the correctness of \Cref{alg:po_hybrid}.
\begin{restatable}{thm}{algopohq}
For any $c > \frac{1}{3}$ Algorithm \ref{alg:po_hybrid} is an $\mathcal{O}(n^{c_0})$-competitive algorithm for eliciting a necessarily Pareto optimal matching in the hybrid-query model.
\end{restatable}
As an easy corollary, this implies that for any $\varepsilon > 0$, we can reach a competitive ratio of $\mathcal{O}(n^{\frac{1}{3} + \varepsilon})$.
Further, we can also show that this competitive ratio is almost asymptotically optimal by showing that no online algorithm can be $o(n^\frac{1}{3})$-competitive. 
\begin{restatable}{thm}{lowerboundhqpo}
There is no online algorithm in the hybrid-query model for computing a necessarily Pareto optimal matching with a competitive ratio of $o(n^\frac{1}{3})$.
\end{restatable}
This of course still leaves the possibility of an online algorithm with a competitive ratio of $\Theta(n^\frac{1}{3})$.

\section{Rank-Maximal Matchings}
In this section, we turn to the problem of eliciting rank-maximal matchings. In their paper, \citet{HMSS21} showed that no online algorithm for the rank-maximal matching problem in the next-best query setting can be better than $\frac{4}{3}$ competitive. Here, we give an algorithm that is $\frac{3}{2}$-competitive and improve the lower bound of \citet{HMSS21} to $\frac{3}{2}$ as well, thus showing that the competitive ratio of our algorithm is tight. However, before defining this algorithm, we first need to recap some results from the classical work of \citet{IKMMP06}. First, we recall the definition of the so-called Dulmage–Mendelsohn decomposition.
\begin{restatable}[\citet{IKMMP06}, \citet{Manl13a}]{thm}{dulmage}
Given a bipartite graph $G = (V,E)$ there exists a partition of $V$ into three sets $\cO, \cE, \cU$ such that 
\begin{itemize}
    \item Any maximum matching only contains edges between $\cU$ and $\cU$ as well as between $\cE$ and $\cO$.
    \item Any maximum matching matches every vertex in $\cU$ and in $\cO$. 
    \item The cardinality of a maximum matching is $\lvert \cO\rvert + \frac{1}{2}\lvert \cU \rvert$.
    \item There is no edge between $\cE$ and $\cE$ as well as between $\cE$ and $\cU$.
\end{itemize}
\label{thrm:dmcomp}
\end{restatable}
Further, it is shown by \citet{IKMMP06} that $\cE$ is the set of vertices which can reach an unmatched vertex in any maximum matching with an alternating path of even length, vertices in $\cO$ can reach an unmatched vertex with an alternating path of odd length, and vertices in $\cU$ cannot reach any unmatched vertex using an alternating path. Using the Dulmage–Mendelsohn decomposition \citet{IKMMP06} defined the following iterative algorithm for finding a rank-maximal matching. In each iteration $i$ the algorithm maintains a graph $G_i = (A \cup H, E_i)$ and a matching $M_i$ in $G$. It is initialized with $E_0 = \emptyset, M_0 = \emptyset, \cO_0 = \emptyset = \cU_0$ and $\cE_0 = A \cup H$. 
For each $i = 1, \dots, n$ the algorithm
\begin{itemize}
    \item adds all edges of rank $i$ or less that have not been deleted yet to $E_i$ and computes a maximum matching $M_i$ in $G_i = (A \cup H, E_i)$ by augmenting $M_{i-1}$;
    \item computes a Dulmage–Mendelsohn decomposition $\cU_i, \cE_i, \cO_i$ for $G_{i}$
    \item deletes all edges incident to a node in $\cU_{i}$ and $\cO_{i}$ with a rank greater or equal to $i$, as well as all edges connecting either two nodes in $\cO_{i}$ or a node in $\cO_{i}$ with a node in $\cU_{i}$;
\end{itemize}
The key observations of \citet{IKMMP06} are now that 
\begin{restatable}{lem}{rmchar}
\begin{itemize}
    \item Every rank-maximal matching in the instance restricted to top-$k$ preferences is a maximum matching in $G_k$.
    \item Every $M_k$ is a rank-maximal matching in the instance restricted to top-$k$ preferences.
\end{itemize}
\label{lem:rm_char}
\end{restatable}
Using these observations, we can now simulate the algorithm of \citet{IKMMP06} by tracking the set of forbidden edges $F$, \ie the set of edges deleted in the third step of \cite{IKMMP06} in each iteration, and by not asking any queries to an agent who has been in $\cU_i$ or $\cO_i$ for any $i \in [n]$. Please refer to \Cref{fig:alg_executed} for an example of \Cref{alg:rm_next_best} being executed.
\begin{algorithm}[tb]
\caption{Elicitation algorithm for rank-maximal matchings in the next-best model}
\label{alg:rm_next_best}
\textbf{Input}: Set of agents $A$, set of houses $H$.\\
\textbf{Output}: A necessarily rank-maximal matching 

\begin{algorithmic}[1] 
\STATE $U \gets A$ \COMMENT{set of unfinished agents}
\STATE $V \gets H$ \COMMENT{set of available houses}
\STATE $E \gets \emptyset$, $M \gets \emptyset$, $F \gets \emptyset$
\IF{$\lvert A \rvert = 2$}
\STATE let $h_1 = \mathcal{Q}(a_1)$, return $\{\{a_1, h_1\}, \{a_2, h_2\}\}$
\ENDIF
\FOR{$i$ in $1, \dots n-1$}
\FORALL{$a \in U$}
\STATE $h\gets \mathcal{Q}(a)$
\IF{$h \in V$ and $\{a,h\} \notin F$}
\STATE $E \gets E \cup \{\{a,h\}\}$
\ENDIF
\ENDFOR
\STATE augment $M$ to be a maximum matching in $(A \cup H, E)$
\STATE compute Dulmage-Mendelsohn decomposition $\cU, \cE, \cO$ for $M$
\STATE If an agent $a \in A$ is in $\cU$ or $\cO$ remove $a$ from $U$
\STATE If a house $h \in H$ is in $\cU$ or $\cO$ remove $h$ from $V$ 
\STATE Add any edges between $\cO,\cO$ and $\cO, \cU$ to $F$ and remove them from $E$
\ENDFOR
\STATE return $M$
\end{algorithmic}
\end{algorithm}
\begin{figure}
    \centering
    \includegraphics[scale = 0.84]{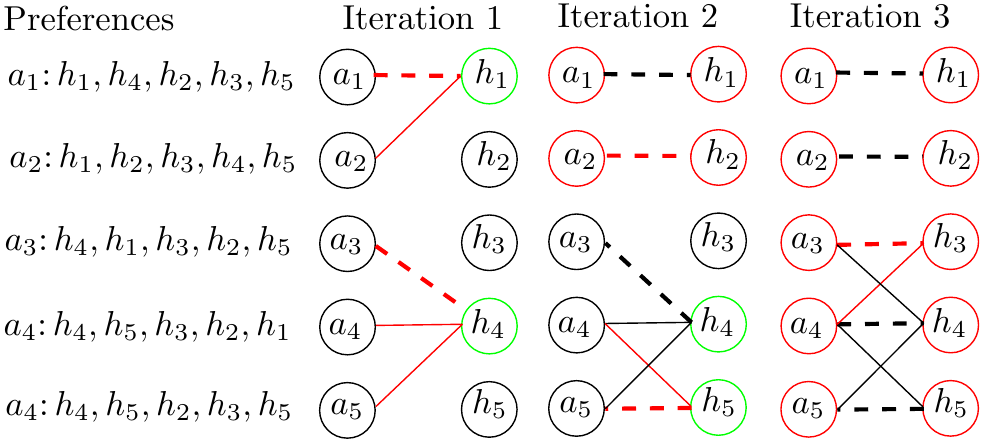}
    \caption{Exemplary run of \Cref{alg:rm_next_best}. The edges displayed are the edges in $E$ at the end of each iteration. Vertices in black are in $\cE$, vertices in green are in $\cO$ and vertices in red are in $\cU$. Edges in red are the edges added in each iteration and dashed edges are matched by $M$ in each iteration(Of course other matchings would also be valid).}
    \label{fig:alg_executed}
\end{figure}
To get a good bound on the competitive ratio of \Cref{alg:rm_next_best} we observe a simple lemma based on the Dulmage-Mendelsohn decomposition and its relation to preferences of agents in a necessarily rank-maximal matching. For a given house allocation instance and agent $a \in A$ let $r_a \coloneqq \min_{i \in [n]} a \notin \cE_i$. It is easy to see that for any instance, our algorithm asks exactly $r_a$ queries to $a$. We can now show that the optimal algorithm must ask at least $r_a-1$ queries to agent $a$ if this agent is matched to a preference the agent has revealed.
\begin{restatable}{lem}{rabound}
Given a house allocation instance $(A,H,\succ)$ partial preferences $\succ'$ in the next-best query model and necessarily rank-maximal matching $M$ for $\succ'$ it has to hold that $\lvert \rev(a) \rvert \ge r_a -1$ for all agents matched to a revealed preference by $M$.
\label{lem:rabound}
\end{restatable}
\begin{proof}
Towards a contradiction, we assume that there is some agent $a \in A$ with $M(a) \in \rev(a)$ and $\lvert \rev(a) \rvert < r_a -1$. Then since $M$ is rank-maximal in all completions of $\succ'$ and thus also in $\succ$ we know that $a \in \cE_{r_a-1}$. Hence, there has to be an alternating path of even length $\rho \coloneqq a = a_1, M(a_1), a_2, \dots, M(a_{k-1}), a_k$ in $G_{r_a-1}$ from $a$ to an unmatched agent $a_k$. Since $a_k$ is unmatched in $G_{r_a-1}$ by $M$ we know that $\rank(a_k, M(a_k)) > r_a-1$ in $\succ$.

Further, assume that there is some $a_i$ with $i>1$ such that $\rank(a_i, M(a_{i-1})) > \rank(a_{i-1}, M(a_{i-1}))$. Then, since $a_{i-1}$ must be in $\cE_j$ for all $j \in [r_a-1]$ due to $\rho$ , we know that $M(a_{i-1})$ must be in $\cO_{\rank(a_{i-1}, M(a_{i-1}))}$ which in turn implies that $\{a_i, M(a_{i-1})\}$ is not an edge in $G_{r_a-1}$. Thus, $\rank(a_i, M(a_{i-1})) \le \rank(a_{i-1}, M(a_{i-1}))$ has to hold in $\succ$. Further, we can extend $\succ'$ in such a way that $\rank(a_1, M(a_k)) \le r_a -1$, while keeping all other preferences according to $\succ$. For these preferences, we can augment $M$ with the path $$a_1, M(a_k), a_k, M(a_{k-1}), a_{k-1}, \dots a_2, M(a_1)$$ and get a matching that rank-dominates $M'$ since $\rank(a_i, M(a_{i-1})) \le \rank(a_{i-1}, M(a_{i-1}))$, as well as $\rank(a_1, M(a_k)) \le r_a-1$, while previously $\rank(a_k, M(a_k)) > r_a - 1$. Therefore, no agent $a\in A$ matched to a revealed preference can have less than $r_a-1$ of their preference revealed. 
\end{proof}
This simple lemma is in fact already sufficient to get a constant upper bound of at most $3$ on the competitive ratio of \Cref{alg:rm_next_best}. To see this, consider the following. We know from Lemma 5 and the fact that at most one agent can be matched to an unrevealed preference and that all but one agent are asked at least $\min(r_a-1,1)$ queries by the optimal algorithm and exactly $r_a$ queries by our algorithm. Further, the agent matched to the unrevealed preference is asked at most $n-1$ queries by our algorithm. Let $A'$ be the set of agents matched to revealed preferences. Then, we can bound the competitive ratio by $$\frac{\sum_{a \in A'}(r_a) + n-1}{\sum_{a \in A'}(\min(r_a-1,1))} \le \frac{\sum_{a \in A'}(r_a + 1)}{\sum_{a \in A'}(\min(r_a-1,1))}.$$ Since for any $a\in A'$ the term $\frac{r_a + 1}{\min(r_a-1,1)}$ is at most $3$, this implies the upper bound of $3$ on the competitive ratio. However, we can improve upon this and can even show that the algorithm is $\frac{3}{2}$-competitive. 
\begin{restatable}{thm}{rmnextbest}
Algorithm \ref{alg:rm_next_best} is a $\frac{3}{2}$-competitive algorithm for eliciting a rank-maximal matching in the next-best query model.
\label{thrm:32-comp}
\end{restatable}
Further, we can show that this bound is tight (up to sub-constant factors) by showing that for any $\varepsilon > 0$, no online algorithm can have a competitive ratio better than $\frac{3}{2}-\varepsilon$.
\begin{figure}
    \centering
    \includegraphics[scale = 0.65]{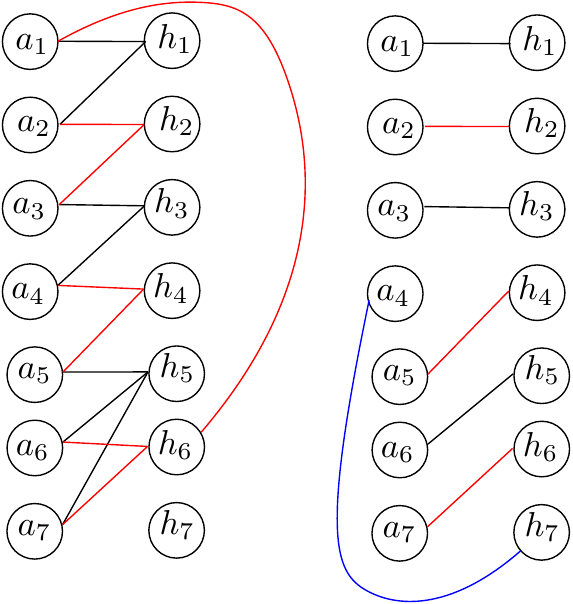}
    \caption{Construction of \Cref{thrm:lower_rm_nb} for $k = 3$ and special agent $a_4$. The basic preferences for ranks $1$ and $2$ are shown on the left, with rank $1$ edges in black and rank $2$ edges in red and the edge of the special agent in blue. The matching on the right, is the rank-maximal matching.}
    \label{fig:examples_lb}
\end{figure}
\begin{restatable}{thm}{lowerrmnb}
No online algorithm in the next-best query model can achieve a competitive ratio better than $\frac{3}{2}-\varepsilon$ for any $\varepsilon > 0$.
\label{thrm:lower_rm_nb}
\end{restatable}
\begin{proof}
Let $k \ge 2$ be an integer and consider the following instance with $n = 2k+1$. We start off with the \emph{basic preferences} of the agents. Here for any $i \in[k]$ (with the assumption that $1-1 = k$), we assume that the basic preference of agent $a_{2i-1}$ is $h_{2i-1} \succ h_{2 (i-1)} \succ \dots \succ h_{2k+1}$ and the basic preference of agent $a_{2i}$ is $h_{2i-1} \succ h_{2i} \succ \dots \succ h_{2k+1}$ with the preferences between the second and the last house being arbitrary. Further, the basic preference of agent $a_{2k+1}$ is $h_{2k-1} \succ h_{2k} \succ \dots \succ h_{2k+1}$ (just like for agent $a_{2k}$). For our adversarial instance, we assume that all but one agent have their basic preference, with the one \emph{special agent} instead listing $h_{2k+1}$ as their third preference. It is easy to see that in such an instance, a rank-maximal matching must match $k$ agents to their first choice, $k$ agents to their second choice and the special agent to $h_{2k+1}$, since the special agent is the only one not listing $h_{2k+1}$ last. The existence of such a matching easily follows by matching the special agent $a_i$ to $h_{2k+1}$, any agent $a_j$ with $j < i$ to $h_j$ and any agent $a_j$ with $j > i$ to $h_{j-1}$. This is a valid matching and matches exactly $k$ agents to their first choice, namely all agents with an odd index that is smaller than $i$ and all agents with an even index that is larger than $i$, $k$ to their second choice and $a_i$ to their third choice, thus being rank-maximal.  For an example of this construction, we refer to \Cref{fig:examples_lb}.

Hence, the optimal algorithm can ask 2 queries to each non-special agent and 3 queries to the special agent and can thus elicit an NRM matching. An adversary on the other hand can simply reveal other houses than $h_{2k+1}$ when asked for the third house of any agent until the last agent is asked. Thus, any online algorithm needs to ask every agent at least $3$ queries. 

Therefore, the competitive ratio of any online algorithm has to be at least $\frac{3(2k +1)}{2(2k+1) +1} = \frac{3}{2} - \frac{3}{8k+6}$ and thus we get that no online algorithm can be $\frac{3}{2}-\varepsilon$-competitive for any $\varepsilon > 0$.
\end{proof}
With some slight adjustments to the adversary, the same construction also works for the hybrid-query and set-compare model, thus also implying a lower bound of $\frac{3}{2}$ in both models. 
\begin{restatable}{cor}{lowerboundothers}
In both the set-compare and hybrid query models, there is no online algorithm with a competitive ratio of $\frac{3}{2}-\varepsilon$ for any $\varepsilon > 0$.
\end{restatable}
We also construct an algorithm, which decides in polynomial time whether a necessarily rank-maximal matching exists, thus standing in contrast to the $\np$-hardness of the same decision problem for Pareto optimal matchings.
\begin{restatable}{thm}{hybridrmpoly}
Given partial preferences $\rank'$ in the hybrid-query model, it can be decided in polynomial time whether a necessarily rank-maximal matching exists and whether a given matching $M$ is necessarily rank-maximal.
\end{restatable}
Finally, we modify \Cref{alg:rm_next_best} to get an algorithm with a constant competitive ratio in the hybrid-query model.
\begin{restatable}{thm}{algormhq}
There exists a $6$-competitive algorithm for eliciting a rank-maximal matching in the hybrid query model.
\end{restatable}
However, we were not able to find an online algorithm for the set-compare setting achieving a sublinear competitive ratio for eliciting a rank-maximal matching.

\section{Discussion}
There are multiple open questions and possible future research directions which can be derived from this work. 

Firstly, there are still gaps left between the upper bounds and lower bounds we showed in this paper. Most importantly, it would be very interesting to find out whether there is an algorithm with a constant (or even sublinear) competitive ratio for eliciting an NRM matching in the set-compare model. Besides this, the complexity of determining whether a matching is NRM is also open in the set-compare model. 

Secondly, there are several other notions of optimality left to explore, for instance \cite{HKMM16} fair matchings or the general class of profile-based matchings \cite{KIMS14} encompassing both fair and rank-maximal matchings. Of course, it might also be interesting to study further querying models or models of partial preferences.

\section{Acknowledgments}
This work was supported by the Deutsche Forschungsgemeinschaft under
grant BR 4744/2-1 and Graduiertenkolleg ``Facets of Complexity'' (GRK 2434). Further, I want to thank the reviewers of AAAI, Markus Brill, Felix Brandt, Jonas Israel, and Ulrike Schmidt-Kraepelin for their helpful comments and suggestions.
\bibliography{pop}
\newpage

\appendix
\section{Missing proofs}
Before giving the missing proofs of the theorems from the main body of work, we give a proof sketch that serial dictatorship mechanism can be efficiently simulated in the set-compare model, even if the size of the sets one can query is upper bounded. We define the set-compare-$k$ query model, as the set-compare query model with the additional requirement that for any query $\mathcal{Q}(a, H')$ it has to hold that $\lvert H' \rvert \le k$.
\begin{proposition}
For any $k > 2$ there is a 1-competitive algorithm in the set-compare-$k$ query model for eliciting an NPO matching.
\end{proposition}
\begin{proofsketch}
We again iteratively simulate the serial dictatorship mechanism. We start off by setting $M = \emptyset$. Then in iteration $i$, we can find the unmatched house agent $a_i$ ranks the highest using $\lceil \frac{n-i}{k-1} \rceil$-queries, by at first choosing an arbitrary subset of size $k$, asking $a_i$ for their top-choice in this subset, then taking $k-1$ other unmatched houses, together with the previous top-choice and so on. The final output of this is guaranteed to be the highest ranked, unmatched house of $a_i$. We can thus add this house and $a_i$ to $M$ in accordance with the serial dictatorship mechanism. Further, we note that we can also generalize \Cref{lem:npo_serial} to the set-compare setting. Thus, for the optimal algorithm there also has to be a permutation of the agents $\sigma$ such that the resulting matching of the optimal algorithm $M'$ is the result of the serial dictatorship mechanism on any preference extension and $\sigma$. Thus, agent $\sigma(1)$ must rank $M(\sigma(1))$ better than all other houses, $\sigma(2)$ must rank $M(\sigma(2))$ better than all other houses except for $M(\sigma(1))$ and so on. It is therefore easy to see that agent $\sigma(i)$ must have been asked at least $\lceil \frac{n-i}{k-1} \rceil$ queries, hence showing that our algorithm is $1$-competitive.
\end{proofsketch}
We note that this is only a proof sketch, with the result not being part of the main work. We defer a proper study of this setting, also with regard to rank-maximality, to future work.

Following this, we give the proofs of the remaining missing theorems from the main body of the paper.
\npcmqnpo*
\begin{proof}
To show membership in $\np$ we observe that the problem of checking whether a matching is NPO is solvable in polynomial time, as shown in \Cref{thrm:alg_npo_mixed}. Thus any NPO matching $M$ is enough as a polynomial size witness and thus the problem is in $\np$.
For our reduction, we reduce from \textsc{(2,2)-e3-sat}. Let $X = \{x_1, \dots, x_n\}$ be the set of variables and $\mathcal{C} = \{C_1, \dots, C_m\}$ be the set of clauses in our \textsc{(2,2)-e3-sat} instance. Each variable $x_i \in X$ appears in exactly two clauses in negative form and two clauses in positive form. We assume that we have some ordering over these clauses such that we can identify one of the clauses as being the first clause $x_i$ (or $\overline{x}_i$ respectively) and the other one as the second one $x_i$ appears in.

In our reduction, we create the following houses:
\begin{itemize}
    \item for each variable $x_i \in X$ we create two \emph{selection houses} $s_i^1$ and $s_i^2$. Further, we create two positive \emph{variable houses} $h_i^1$ and $h_i^2$ and two negative variable houses $\overline{h}_i^1$ and $\overline{h}_i^2$. If $C$ is the clause $x_i$ appears in the first time, we say that $h_i^1$ belongs to $C$ and similarly $h_i^2$ belongs to $C$ if $C$ is the second clause $x_i$ appears in. The equivalent notation is also used for $\overline{x}_i$ with the negative variable houses.
    \item for each clause $C \in \mathcal{C}$ we add a \emph{clause house} $h_C$;
    \item we add $\frac{2}{3}n+2$ \emph{dummy houses} $d_1, \dots ,d_{\frac{2}{3}n}, d$, and $d'$.  
\end{itemize}
Note that due to the nature of our problem $3m = 4n$ holds and therefore the total number of houses is $6n + m + \frac{2}{3}n+2 = 6n + \frac{4}{3}n + \frac{2}{3}n+2 = 8n+2$.
Next we come to the agents and their revealed preferences. 
\begin{itemize}
    \item For each variable $x_i \in X$ we add four agents, two \emph{positive agents} $x_i^1$ and $x_i^2$ and two \emph{negative agents} $\overline{x}_i^1$ and $\overline{x}_i^2$. The positive agent $x_i^1$ has revealed that $\rank'(x_i^1, s_i^1) = 1,\rank'(x_i^1, s_i^2) = 2, \rank'(x_i^1, h_i^1) = 4, \rank'(x_i^1, h_i^2) = 5$. Further, for each other house except for $\overline{h}_i^1$ and $\overline{h}_i^2$ the agent has a revealed rank between $6$ and $8n+1$, with variable houses being ranked before dummy houses and dummy houses before clause houses. The internal ranking in these house classes is arbitrary, except that we require every house from $d_1, \dots, d_{\frac{2}{3}n}$ to be before $d$ and $d'$. Therefore, the only unrevealed ranks are $3$ and $8n+2$ and the only unrevealed houses are $\overline{h}_i^1$ and $\overline{h}_i^2$. The revealed preferences for the second positive agent $x_i^2$ are the same. For $\overline{x}_i^1$ and $\overline{x}_i^2$ we have the same preferences, except that they have $h_i^1$ and  $h_i^2$ unrevealed and $\overline{h}_i^1$ and $\overline{h}_i^2$ revealed.
    \item Next, for each clause $C \in \mathcal C$ we add three agents $C_1, C_2$, and $C_3$. Let $h_1, h_2,$ and $h_3$ be the three variables houses belonging to $C$. 
    Then we  assume that the three agents have revealed their top $\frac{2}{3}n + 4$ ranks such that that they induce the preferences
    
    $\begin{array}{cc}
         & C_1 \colon h_1 \succ h_2 \succ h_3 \succ h_C \succ d_1 \succ \dots \succ d_{\frac{2}{3}n} \\
         & C_2 \colon h_2 \succ h_3 \succ h_1 \succ h_C \succ d_1 \succ \dots \succ d_{\frac{2}{3}n} \\
         & C_3 \colon h_3 \succ h_1 \succ h_2 \succ h_C \succ d_1 \succ \dots \succ d_{\frac{2}{3}n} \\
    \end{array}$
    
    with the preferences over all other houses being unrevealed.
    \item Finally, we add two \emph{dummy agents} $a^1_d, a^2_d$ who have revealed rank such that the preference list from position $m$ to $8n+1$ is $d \succ d' \succ h_1^1 \succ \overline{h}_1^1 \succ s_1^1 \dots \succ \overline{h}_n^2 \succ s_n^1 \succ s_n^2 \succ d_1 \succ \dots \succ d_{\frac{2}{3}n}$. Thus, the only unrevealed houses of $a_d^1$ and $a_d^2$ are the clause houses, one of which is listed last while the other $m-1$ clause houses are a prefix of the preference list.
\end{itemize}
    Note, that there are $4n + 3m + 2 = 8n + 2$ agents, thus fulfilling the condition that the number of agents and houses is the same. 
    
    \paragraph{$\Rightarrow$} First, we assume that our \textsc{(2,2)-e3-sat} instance is satisfiable and that $\Phi$ is a satisfying assignment. We now show how to construct an ordering of the agents, such that the serial dictatorship method with regard to that ordering is consistent with any possible preference completion. First, for any variable $x_i$ if $x_i$ appears in $\Phi$, our serial dictatorship mechanism asks $\overline{x}_i^1$ and then $\overline{x}_i^2$ (who get matched to $s_i^1$ and $s_i^2$), followed by the clause agents who rank $\overline{h}_i^1$ and $\overline{h}_i^2$ first. Afterwards, we ask $x_i^1$ and $x_i^2$ to pick their top choice, which is $h_i^1$ and $h_i^2$ regardless of the preference completion. Now, since $\Phi$ is a satisfying assignment, for any clause $C$ there has to be at least one agent $C_j$ who is currently unmatched. This is due to the fact that the only clause houses matched, are matched to variable houses, belonging to variables not in $\Phi$. For these clause houses, the top-choice which is not matched yet is guaranteed to be $h_C$, since all the variable houses are matched. We can therefore ask a clause agent to be matched to their corresponding clause house. Finally, we can match the remaining clause agents to dummy agents and the dummy agents to $d$ and $d'$, since all clause houses are matched. Since this is a serial dictatorship order that works for any completion of the preferences, the matching is guaranteed to be NPO.
    \paragraph{$\Leftarrow$}
    Let $M$ be an NPO matching. First, assume that $M(a^1_d) \neq d$ and $M(a^1_d) \neq d'$. We will distinguish two cases based on the houses $a_d$ could be matched to. 
    \begin{itemize}
        \item If $a^1_d$ is matched to a clause house $h_C$, there is a completion such that $h_C$ is ranked last by $a^1_d$ and first by $a^2_d$, thus implying that they form a cycle.  
        \item If $a_d^1$ is matched to either a selection, dummy, or variable house, then every preference list except for $a_2^d$ can be completed in such a way that the house matched to $a_d^1$ is in front of $d$ and $d'$ in each preference list. Therefore, there needs to be a cycle containing $a_d^1$ and another agent with both preferring the partner house of the other.  
    \end{itemize}
    Due to $a_d^2$ having the same revealed preferences as $a_d^1$ this also implies that $a_d^2$ is matched to either $d$ or $d'$. 
    
    Next, we assume that there is any variable agent matched to a clause house. However, then this variable agent would prefer $d$ and $d'$ to their current house, while the preferences of $a_d^1$ could be completed to prefer the clause house to $d$ and $d'$ thus inducing a cycle. Thus, we know that every clause house must be matched to a clause agent. Let $C, C' \in \mathcal{C}$ be two distinct clauses and assume that $C_j$ for some $j \in [3]$ is matched to $h_{C'}$. Since $h_{C'}$ as well as $d$ and $d'$ are unrevealed by $C_j$ and since the dummy agent $a_d^1$ could prefer $h_{C'}$ to $d$ and $d'$ we know that $C_j$ cannot be matched to $h_{C'}$ in any NPO matching. Therefore, for any clause $C \in \mathcal{C}$ there needs to be one agent $C_j$ with $M(C_j) = h_{C_j}$. This however implies that at least one variable agent corresponding to a variable from $C_j$ must be matched to their corresponding variable houses. Since the negative houses could be ranked last by any positive agent and vice versa, this implies that this agent matched to a variable house must be positive if the variable house is positive and negative if the variable house is negative. 
    
    Let $\Phi$ be the assignment setting all variables to true for which a positive agent is matched to a variable house and all variables to false for which a negative agent is matched to a variable house. To show that this is a valid definition, assume that there is some variable $x_i$ for which both a negative and a positive agent are matched to a variable house. Then it is easy to see that the preferences of both agents could be completed in such a way that both agents prefer each other's house to their own. Thus, $\Phi$ is a valid assignment. Moreover, we know that for each clause, there is at least one agent, who is also positive if and only if the house is positive, matched to this variable house. Thus, each clause is fulfilled by at least one variable from $\Phi$ and the \textsc{(2,2)-e3-sat} instance is satisfiable.
\end{proof}
\scnponpc*
\begin{proof}
For this, we give a simple reduction from the problem of determining whether an NPO matching exists in the hybrid-query setting. Assume that we are given partial preferences $\rank'$ in the hybrid-query model. For a given agent $a \in A$ we call a revealed house $h \in H$ a prefix house of $a$, if all ranks from $1$ to $\rank(a, h)$ are revealed. It is easy to see that no unrevealed house can be ranked better than a prefix house, while all unranked houses could be ranked better than a non-prefix house.

For our reduction, we now create partial preferences $\succ'$ in the set-compare model, by taking the partial order for each agent which ranks 
\begin{itemize}
    \item every prefix house better than every unrevealed house and every house with a worse rank;
    \item every revealed non prefix house better than every revealed house with a worse rank.
\end{itemize} 

Now assume we are given a matching $M$. We show that $M$ is NPO in the hybrid-query preferences if and only if $M$ is NPO in the constructed set-compare preferences.

First, assume that $M$ is not necessarily Pareto optimal in the hybrid-query preferences. Then there is a cycle $a_1, \dots, a_k = a_1$ with agent $a_i$ preferring $M(a_{i+1})$ to $M(a_i)$ in some preference extension of $\rank'$. Therefore, we know that $M(a_i) \succ'_i M(a_{i+1})$ does not hold in $\succ'$. Since $\succ'_i$ is a partial order, we can extend $\succ'_i$ to $\succ_i$ in such a way that $M(a_{i+1})\succ_i M(a_i)$. By doing this for all agents, we see that $M$ is also not necessarily Pareto optimal in $\succ'$. Similarly, if $M$ is not NPO in $\succ'$ there is some cycle  $a_1, \dots, a_k = a_1$ with agent $a_i$ preferring $M(a_{i+1})$ to $M(a_i)$ in some preference extension of $\succ'$. Thus, for each agent $a_i$ on the cycle there is some extension of $\rank'$ with $\rank'(a_i, M(a_i)) \ge \rank'(a_i, M(a_{i+1}))$ which in turn shows that $M$ is also not NPO for $\rank'$. 

Thus, we can reduce to the problem of determining whether a matching is NPO in the set-compare model and hence this problem is $\np$-complete as well. 
\end{proof}
\increasing*
\begin{proof}
We show this by induction. Let $j > 0$ be given and assume that $c_j > \frac{1}{3}$. Further, let $\varepsilon = \frac{1}{3} - c_0$. Then,  \begin{align*}
    c_{j+1} = \frac{3c_{j} + 1}{2} + c_0 -1 = c_j + \frac{c_j + 1}{2} + \frac{1}{3} + \varepsilon - 1 \\
    \ge c_j + \frac{2}{3} + \frac{1}{3} + \varepsilon -1 = c_j + \varepsilon.
\end{align*}
Thus, it holds that $c_{j+1} \ge c_j + \varepsilon$ and the sequence is monotonically increasing.
\end{proof}
\algopohq*
\begin{proof}
To show this bound, let $(A, H, \succ)$ be a given house allocation instance.

First, we observe that the matching $M$ is indeed Pareto optimal. For this, let $A'$ be the set of agents matched in the first 'for loop'. Since $M$ is Pareto optimal when restricted to $A'$, there is some permutation $\sigma'$ of $A'$ such that $SD_\succ(\sigma') = M$ when restricted to $A'$. We now create a permutation $\sigma$ of $A$ by first copying $\sigma'$ and for any $j \in \left[\lvert A' \rvert, n\right]$ setting $\sigma(j)$ to be the agent matched the $j$th time by our algorithm. Now, running the serial dictatorship mechanism on $\sigma$ results in any agent in $A'$ being matched to their respective partner in $M$, due to $\sigma$ consisting of $\sigma'$. Further, Algorithm \ref{alg:po_hybrid} matches any agent in the second for loop to the unmatched house they prefer the most, thus resulting in $SD_\succ(\sigma) = M$.

We show that Algorithm \ref{alg:po_hybrid} is $\mathcal{O}(n^{c_0})$-competitive by considering the maximum value of $j$. 
\begin{itemize}
    \item As our base case, if the maximum value of $j$ is $0$, we get that for $i = \lceil n^{c_0}\rceil$ we have that $\lvert M\rvert \ge n-\left\lceil n^{(c_j+1)/2}\right\rceil$, and our algorithm at most takes $n \lceil n^{c_0}\rceil + \left\lceil n^{(c_0+1)/2}\right\rceil\left\lceil n^{(c_0+1)/2}\right\rceil \in \mathcal{O}(n^{1+c_0})$ queries. Since the optimal algorithm uses at least $n-1$ queries, the competitive ratio in this case is in $\mathcal{O}(n^{c_0})$.
    \item Now assume that the maximum value $j$ takes is at least $1$ (and let $j$ be this maximum value) and that the for loop breaks. Then for $i = \lceil n^{c_{j-1}}\rceil$ we know that $M < n - \left\lceil n^{(c_{j-1}+1)/2}\right\rceil$. Thus, at least $\left\lceil n^{(c_{j-1}+1)/2}\right\rceil$ agents are matched to a house with a rank worse than $\left\lceil n^{c_{j-1}}\right\rceil$ for any NPO matching. 
    
    Let $M'$ be the NPO matching produced by the optimal algorithm with partial preferences $\rank'$ and $\sigma$ the corresponding permutation of $A$ such that $SD_{\succ'}(\sigma) = M'$ for all extensions $\succ'$ of $\rank'$. By definition of the serial dictatorship mechanism it is easy to see that for any $i \in [n]$ and $a \in A$ we have that $\sigma(i) = a$ implies $\rank(a, M'(a)) \le i$. Therefore, and by the fact that $\left\lceil n^{(c_{j-1}+1)/2}\right\rceil \ge \lceil n^{c_{j-1}}\rceil$, we get that $$\min(\rank(a, M'(a)), n-\sigma^{-1}(a)) \ge \frac{1}{2}\lceil n^{c_{j-1}}\rceil$$ for at least $\frac{1}{2}\left\lceil n^{(c_{j-1}+1)/2}\right\rceil$ many agents, since $\sigma^{-1}(a) \le n - \frac{1}{2}\left\lceil n^{(c_{j-1}+1)/2}\right\rceil$ for at least $\frac{1}{2}\left\lceil n^{(c_{j-1}+1)/2}\right\rceil$ many $a \in A$ with $\rank(a, M'(a))\ge \lceil n^{c_{j-1}}\rceil$. Thus, by Lemma \ref{lem:lower_bound_po} the optimal algorithm needs at least \[\Omega\left(\left\lceil n^{(c_{j-1}+1)/2}\right\rceil\lceil n^{c_{j-1}}\rceil\right) \in \Omega\left(n^{(3c_{j-1}+1)/2}\right)\] queries.
    Since our algorithm uses \begin{align*}
        \mathcal{O}(n n^{c_j} + n^{(c_j+1)/2}n^{(c_j+1)/2}) = \mathcal{O}(n^{c_j+1}) = \\  \mathcal{O}(n^{1+ (3c_{j-1} + 1)/2 + c_0 -1}) = \mathcal{O}(n^{ (3c_{j-1} + 1)/2}n^{c_0})
    \end{align*} queries, we get a competitive ratio of $\mathcal{O}(n^{c_0})$ in this case.
    \item Finally, if the for loop never breaks, we know that our algorithm uses $\mathcal{O}(n^2)$ many queries. Further, we can assume that $n^{c_o - \frac{1}{3}} \ge 2$ and $n^{c_0} \ge 1$ since we are only dealing with asymptotics. Then following the calculations of Lemma \ref{lem:increasing} we get that \begin{align*}
        \lceil n^{c_{j+1}} \rceil \ge\lceil n^{c_{j} + c_0 - \frac{1}{3}} \rceil \\
        \ge \lceil 2n^{c_{j}} \rceil > \lceil n^{c_{j}} \rceil + 1. 
    \end{align*}
    Therefore, the sequence $(\lceil n^{c_j}\rceil)_{j \in \mathbb{N}}$ is monotonically increasing as well.
    Further, we know that for the maximum $j$ it has to hold that $(3c_{j-1} + 1)/2 + c_0 -1 > 1$ since the sequence is monotonically increasing by Lemma \ref{lem:increasing} and otherwise the algorithm would enter the case for $i = n^1$. This implies that $c_0 > 2 - (3c_{j-1} + 1)/2$. By the same argument as in the last case, we can lower bound the number of queries asked by the optimal algorithm by $\Omega\left(n^{(3c_{j-1}+1)/2}\right)$ and thus get a competitive ratio of $\mathcal{O}(n^{2 - (3c_{j-1} + 1)/2}) \in \mathcal{O}(n^{c_0})$.
\end{itemize}
\end{proof}
\lowerboundhqpo*
\begin{proof}
For this, we construct a family of instances such that any correct online algorithm must ask at least $\Omega(nn^\frac{1}{3})$ queries, while an optimal offline algorithm can do it using $\mathcal{O}(n)$ queries. For simplicity, we assume that $n^\frac{1}{3}$ is an integer.  We divide the set of agents into three sets
\begin{itemize}
    \item a set of \emph{first-choice} agents $A_1$, consisting of $n - 2n^\frac{2}{3}$ agents,
    \item a set of \emph{second-choice} agents $A_2$ consisting of $n^\frac{2}{3}$ agents,
    \item and a set of \emph{special agents} $A'$ of size $n^\frac{2}{3}$.
\end{itemize}  Similarly, we also divide our set of houses into a set of \emph{first-choice} houses $H_1$ of size $n - n^\frac{2}{3}$ and a set of \emph{special houses} $H_2$ of size $n^\frac{2}{3}$. The preferences of the first-choice agents are such that each first-choice agent has a unique first-choice house as their first-choice and lists all special houses at the end of their preference list with the rest of their preference list being completed arbitrarily. Every special agent also has a unique first-choice house as their first choice and lists a special house somewhere in their first $n^\frac{1}{3}$ preferences. All other special houses are listed in last $n^\frac{2}{3}$ as well. Finally, all second-choice agents, share their first-house with some other agent and list the first choice of a unique special agent somewhere in their first $n^\frac{1}{3}$ preferences. Our family of instances now consist of all preferences inducing such profiles. 
Our adversary now works as follows:
\begin{itemize}
    \item For any agent asked to reveal their first-choice house, we just reveal consistently any 'first-choice' house that has not been revealed twice yet.
    \item For any agent asked to reveal a house in their first $n^\frac{1}{3}$ preferences, we always reveal a first-choice house while still possible.
    \item For any agent asked to reveal the rank of a special house, we always return the lowest consistent rank in the last $n^\frac{2}{3}$ ranks, while this is still possible. 
    \item The rest of the preferences are revealed consistently.
\end{itemize}
We claim that any online algorithm working against this adversary must ask at least $\Omega(nn^\frac{1}{3})$ queries. Let $A'$ be the set of agents matched to special houses. For any $a' \in A'$ there are now two possibilities. If $a' \in A_2$, so $a_2$ is a special agent, and $a'$ is matched to one of their first $n^\frac{1}{3}$ choices, then we know that these preferences could not be completed consistently by placing the agent in the last $n^\frac{2}{3}$ houses of $a'$. Thus, this agent could not have been a first-choice agent. This, however, easily implies that each first choice agent must have been asked at least $n^\frac{1}{3}$ queries, since otherwise the preferences could be consistently completed with that agent being a special agent. Hence, this would already imply that our online algorithm asked $\Omega(n^\frac{4}{3})$ queries.

Therefore, the only choice remaining is all special houses being matched to agents who list them in their last $n^\frac{2}{3}$ preferences. Thus, using Lemma \ref{lem:lower_bound_po} we can deduce that any online algorithm must have asked at least $\frac{1}{2}n^\frac{2}{3}$ queries to at least $\frac{1}{2}n^\frac{2}{3}$ of them, therefore also implying that any online algorithm must ask at least $\Omega(n^{1+\frac{1}{3}})$ queries in this case. An algorithm having access to all preferences on the other hand can ask any top-choice agent for their top-choice, any special agent until it reaches the special house of this agent and any second-choice agent until it reaches the unique top-choice house of a special agent in ranks in the top $n^\frac{1}{3}$ houses. After this the preferences are top-$k$ preferences and a matching of size $n$ is possible. By \citet{HMSS21} this implies that an NPO matching exists. Since it can be found using $\mathcal{O}(n + n^\frac{1}{3}n^\frac{2}{3}) = \mathcal{O}(n)$ queries, this shows that no algorithm can be $o(n^\frac{1}{3})$-competitive. 
\end{proof}

\rmnextbest*
\begin{proof}
To see that the matching $M$ computed by Algorithm \ref{alg:rm_next_best} is indeed rank-maximal, we observe that the graph $(A \cup H, E)$ after finishing the elicitation step in iteration $i$ is equivalent to the graph $G_i$. Every agent only elicits preferences until they are in either $\cU$ or $\cO$ which is equivalent to the deletion of all preferences of a higher rank. Further, by deleting all edges between $\cO, \cO$ and $\cO, \cU$ and by adding them to $F$ we ensure that these edges cannot be used or added to the graph. Finally, removing houses from $V$ once they are in $\cU$ or $\cO$ is again equivalent to the deletion step in the original algorithm. Thus, by Lemma \ref{lem:rm_char} the matching $M$ is rank-maximal. To show that the algorithm is also $\frac{3}{2}$-competitive, we show that almost all agents need to be queried until they would be a part of either $\cU$ or $\cO$ in the decomposition. In \Cref{lem:rabound} we have already shown that no agent $a\in A$ matched to a revealed preference can have less than $r_a-1$ of their preference revealed. 

Now we turn to agents matched to an agent they have not revealed yet \footnote{Note that for these agents \Cref{lem:rabound} does not hold. It is easy to construct an instance, where an agent with no revealed preferences is matched to their last choice.}. Let $M'$ be any NRM matching with corresponding partial preferences and $a \in A$ an agent matched to a house they have not revealed yet. Now, if either $M'(a) \notin \rev(a')$ or $\rank(a', M'(a)) \neq n$ for any $a' \in A$, we can get a matching that dominates $M'$ in an extension of the preferences by matching $a$ to $M'(a')$ and $a'$ to $M'(a)$. Thus, $\rank(a', M'(a)) = n$ has to hold for any $a' \in A \setminus \{a\}$. Thus, any optimal algorithm leaving matching an agent to a preference this agent has not revealed yet has to ask $(n-1)^2$ queries. On the other hand, we know that our algorithm asks at most $n(n-1)$ queries. Thus, in this special case, we get an upper bound on the competitive ratio of $\frac{n(n-1)}{(n-1)^2} = \frac{n}{n-1}$. Since our algorithm is $1$-competitive for $n\le 2$ this implies a competitiveness of $\frac{n}{n-1} \le \frac{3}{2}$ if the optimal algorithm matches an agent to an unqueried house.

Therefore, we can now assume that $\lvert \rev(a)\rvert \ge \max(r_a - 1, 1)$ for all $a \in A$. Finally, for $h \in H$ let $A_h \coloneqq \{a \in A \mid \rank(a,h) = 1\}$. It is easy to see that for $n > 2$, there need to be at least $2 \lvert A_h \rvert -1$ queries to agents in $A_h$, since otherwise there would be two agents with only their first agent revealed, which would be a contradiction to every agent only being matched to a revealed preference. Of course, this also holds for any non-empty subset of $A_h$. Now let $\opt(A')$ be the number of queries asked to agents in the set $A' \subseteq A$. We thus get that for any $h \in H$ and $A' \subseteq A_h$ it has to hold that 
\[
\opt(A') \ge \max(\sum_{a \in A'}(\max(r_a-1, 1)), 2 \lvert A'\rvert -1)
\]
Since our algorithm asks \[\sum_{a \in A} r_a = \sum_{h \in H} \sum_{a \in A_h} r_a\] queries we get that the competitive ratio is at most \[\frac{\sum_{h \in H} \sum_{a \in A_h} r_a}{\sum_{h \in H} \opt(A_h)}.\] By using the fact that for all $a,b,c,d \in \mathbb{R}_{>0}$ it holds that \[\frac{a+b}{c+d} \le \max(\frac{a}{c}, \frac{b}{d})\] we thus get an upper bound of \[\max_{h \in H} \frac{\sum_{a \in A_h} r_a}{\opt(A_h)}\] on the competitive ratio. For any $h \in H$ we can now distinguish two cases. If $\sum_{a \in A_h} r_a \ge 3\lvert A_h \rvert$ we get \begin{align*}
    \frac{\sum_{a \in A_h} r_a}{\opt(A_h)} \le \frac{\sum_{a \in A_h} r_a}{\max(\sum_{a \in A_h}(\max(r_a-1, 1)), 2 \lvert A_h\rvert -1)} \\\le \frac{\sum_{a \in A_h} r_a}{\sum_{a \in A_h}(\max(r_a-1, 1))} \le \frac{\sum_{a \in A_h} r_a}{\sum_{a \in A_h} r_a - \lvert A_h \rvert} \\\le \frac{\sum_{a \in A_h} r_a}{\sum_{a \in A_h} r_a - \frac{1}{3}\sum_{a \in A_h} r_a} = \frac{\sum_{a \in A_h} r_a}{\frac{2}{3}\sum_{a \in A_h} r_a}  =  \frac{3}{2}.
\end{align*} 

If $\sum_{a \in A_h} r_a < 3\lvert A_h \rvert$ there must be at least one $a \in A_h$ with $r_a = 2$. Let $A_h^2$ be this subset. If $\lvert A_h^2 \rvert \ge 2 $ we can see that 
\begin{align*}
    \frac{\sum_{a \in A_h} r_a}{\opt(A_h)} = \frac{\sum_{a \in A_h\setminus A_h^2} r_a + 2 \lvert A_h^2\rvert }{\opt(A_h \setminus A_h^2) + \opt(A_h^2)  } \\ 
    \le \max(\frac{\sum_{a \in A_h\setminus A_h^2} r_a }{\sum_{a \in A_h\setminus A_h^2} r_a - \lvert A_h\setminus A_h^2 \rvert }, \frac{2 \lvert A_h^2\rvert}{2 \lvert A_h^2\rvert - 1}) \le \frac{3}{2}
\end{align*} since we know that $\sum_{a \in A_h\setminus A_h^2} r_a \ge 3 \lvert A_h\setminus A_h^2 \rvert $.
    If, however $\lvert A_h^2 \rvert = 1$, then there must be exactly one $a \in A_h$ with $r_a = 2$. This implies that either $a \in \cO_2$ or $a \in \cU_2$. Since all other agents in $A_h$ are in $\cE_2$ this shows that $h \notin \cE_2$. Therefore, $G_n$ contains no edge between $a$ and $h$ and any rank-maximal matching matches $a$ to their second ranked house. However, then the optimal algorithm must have asked two queries to agent $a$ and therefore we get an upper bound on the competitive ratio of \begin{align*}
    \frac{\sum_{a \in A_h} r_a}{\max(\sum_{a \in A_h}(\max(r_a-1, 1)), 2 \lvert A_h\rvert -1)}\\ \le  \frac{3\lvert A_h \rvert-1}{ 2\lvert A_h \rvert} \le \frac{3}{2}.
\end{align*}
    Thus our algorithm correctly elicits a necessarily rank-maximal matching with a competitive ratio of $\frac{3}{2}$.
\end{proof}
\lowerboundothers*
\begin{proof}
~\paragraph{Hybrid-query model.} We begin with the hybrid query model and take the same preference structure as in the proof of \Cref{thrm:lower_rm_nb}, \ie at the family of preferences isomorphic to the preference structure in the proof. However, in this case, we need to change the elicitation of the preferences by our adversary slightly. When asked for a first-choice house, we still reveal the preferences according to the basic preferences and when asked for the third preference or the rank of the special house we can either return any first-choice house or the last rank for this special house until the last agent queried. 

For the second-choice houses, however, we need to adjust the adversary. For $i \in [k-1]$ we call the agents $a_{2i-1}$ and $a_{2i}$ the block $A_i$ and similarly the agents $a_{2k-1},a_{2k},$ and $a_{2k+1}$ form the block $A_{k}$. Based on the total preferences, we see that the agents in block $A_i$ share their second choice agents with blocks $A_{i-1}$ and $A_{i+1}$ with indices taken modulo $k$. 

While eliciting preferences, we call two blocks $A_i$ and $A_j$ neighboring if they share a second-choice house. It is easy to see that any way of answering the queries for second-choice houses such that the blocks form a cycle of length $n$ based on the neighborhood relation leads to a preference profile isomorphic to the one specified in \Cref{thrm:lower_rm_nb}. In order to use this, until the final preference of the agents is revealed, if any agent is queried for their second choice agent, the adversary first checks if there is any second-choice house revealed by an agent in a block not reachable by the neighborhood relation. If this is the case, we can safely return this agent as the second-choice. If none such agent exists, we return any second-choice agent, not queried before. For the final agent revealing their second-choice agent, there is only one choice left, which we reveal. 

This model of eliciting preferences leads to a cycle of blocks as described above, and thus to preferences isomorphic to the ones described in \Cref{thrm:lower_rm_nb}. Further, every second-choice agent needs to appear at least once as a second-choice agent for any matching to be NRM.

Now assume that the blocks do not form a path according to the neighbor relation. Then either one second-choice houses is still unrevealed or the last query to an agent would be able to reveal a second-choice house in a different connected component according to the neighbor relation, since these two blocks cannot reach each other. Thus, since the blocks need to form a path, we need to reveal two second-choice houses for at least $k-2$ blocks. Thus, we need to ask at least $2k-4$ second-choice queries. Since we also need to query every first and every third choice, this leads to a competitive ratio of at least $\frac{2k-4 + 2(2k+1)}{2(2k+1) + 1} = \frac{3}{2} - \frac{13}{8k+6}$. Thus, also no algorithm can be $\frac{3}{2}-\varepsilon$-competitive for the hybrid-query model.
~\paragraph{Set-compare model.}
The same construction also works in the set-compare model. If an agent $a_i$ is asked for their best-choice of a set containing all houses (or their top-choice house), we simply return their top-choice house. If asked for their best-choice out of a set only not containing their top-choice house \ie their second choice house, we reveal the preferences according to the construction in the hybrid-query model. Finally, if asked to reveal their best-choice out of a smaller set not containing their first or second choice we simply reveal any house that is not the special house until all agents have been queried for the special house at least once. 

Since we need to know the top-choice of every agent we need to ask every agent for their best-choice. Further, since we need to reveal every second-choice agent at least once as the best-choice out of a set, we need to ask at least $2k-4$ queries as previously shown. Finally, we need to ask every agent at least one query containing neither their first and second choice, but the special house. Thus, the competitive ratio is also lower bounded by $\frac{3}{2} - \frac{13}{8k+6}$ in this case.
\end{proof}
\begin{algorithm}[tb]
\caption{Elicitation algorithm for rank-maximal matchings in the hybrid-query model}
\label{alg:rm_hq}
\textbf{Input}: Set of agents $A$, set of houses $H$.\\
\textbf{Output}: A necessarily rank-maximal matching 

\begin{algorithmic}[1] 
\STATE $U \gets A$ \COMMENT{set of unfinished agents}
\STATE $V \gets H$ \COMMENT{set of available houses}
\STATE $E \gets \emptyset$, $M \gets \emptyset$, $F \gets \emptyset$
\IF{$\lvert A \rvert = 2$}
\STATE let $h_1 = \mathcal{Q}(a_1,1)$, return $\{\{a_1, h_1\}, \{a_2, h_2\}\}$
\ENDIF
\FOR{$i$ in $1, \dots n-1$}
\FORALL{$a \in U$}
\STATE $h\gets \mathcal{Q}(a,i)$
\IF{$h \in V$ and $\{a,h\} \notin F$}
\STATE $E \gets E \cup \{\{a,h\}\}$
\ENDIF
\ENDFOR
\STATE augment $M$ to be a maximum matching in $(A \cup H, E)$
\IF{$M$ matches no house of rank $i$}
\STATE $c \gets c+1$
\ENDIF
\IF{$c \ge \lvert V \rvert$}
\STATE break for-loop
\ENDIF
\STATE compute Dulmage-Mendelsohn decomposition $\cU, \cE, \cO$ for $M$
\STATE If an agent $a \in A$ is in $\cU$ or $\cO$ remove $a$ from $U$
\STATE If a house $h \in H$ is in $\cU$ or $\cO$ remove $h$ from $V$ 
\STATE Add any edges between $\cO,\cO$ and $\cO, \cU$ to $F$ and remove them from $E$
\ENDFOR
\IF{$U \neq \emptyset$}
\FORALL{$h \in V, a \in U$}
\STATE query $\mathcal{Q}(a,h)$
\ENDFOR
\STATE compute rank-maximal matching $M$ with completed preferences
\ENDIF
\STATE return $M$
\end{algorithmic}
\end{algorithm}

\hybridrmpoly*
\begin{proof}
In order to show this, we provide a short necessary condition of necessarily rank-maximal matchings, provided that one exists. Given partial preferences in the hybrid query model, for an agent $a \in A$ let $\ell_a$ be the last rank with no revealed house and $k_a$ be the first rank with no revealed house. Then we define $\max\rank(a,h)$ to be $\rank(a,h)$ if $h \in \rev(a)$ and $\max\rank(a,h) = \ell_a$ otherwise, \ie $\max\rank(a,h)$ is the worst possible rank of $h$ for $a$ in any possible completion. We now show that if an NRM matching $M$ exists it is rank-maximal in the instance $(A,H, \max\rank)$. 

Assume that $M$ is indeed NRM and that there is a matching $M'$ rank-dominating $M$ in $(A,H, \max\rank)$. Then we extend the preferences, by setting $\rank(a, M(a)) = \ell_a$ for all $a \in A$ and by extending all other preferences in any consistent manner. The signature of $M$ under $\rank$ and $\max\rank$ is the same, while the signature of $M'$ under $\rank$ is guaranteed to be at least as large as under $\max\rank$. Thus, $M'$ also rank-dominates $M$ under $\rank$ and $M$ is not NRM. 

For our next step let $(\cE_1, \cO_1, \cU_1), \dots, (\cE_n, \cO_n, \cU_n)$ be the Dulmange-Mendelsohn decompositions of $(A,H, \max\rank)$. Let $a \in A$ be an agent and $h \notin \rev(a)$. If either $a \notin \cE_{\ell_a-1}$ or $h \notin \cE_{\ell_a-1}$ then we know that no rank-maximal matching in $(A,H, \max\rank)$ and thus also no NRM matching can match $a$ and $h$. On the other hand, assume that $a \in \cE_{\ell_a-1}$ and $h \in \cE_{\ell_a-1}$. Then we can extend the preferences in such a way that $\rank(a,h) = k_a$ and for all other $a' \in A$ with $h \notin \rev(a)$ it holds that $\rank(a', h) = \ell_{a'}$. Under these preferences $a \in \cE_{k_a-1}$ and $h \in \cE_{k_a-1}$ still hold. Further, if there was another $a' \in \cE_{k_a-1}$ with $\rank(a', h) = k_a$ this would have implied that $h \notin \cE_{\ell_a-1}$, since we only decreased the ranks of the other agents and since $\ell -1 \ge k_a$. Hence, $a,h \in \cU_{k_a}$ has to hold and they have to be matched. Therefore, an agent $a$ is matched to an unrevealed agent $h$ in an NRM matching if and only if $a \in \cE_{\ell_a-1}$ and $h \in \cE_{\ell_a-1}$. Thus, it is sufficient to calculate the Dulmange-Mendelsohn decompositions of $(A,H, \max\rank)$, match all agents $a \in A$ to unrevealed houses $h \in H$ if $a \in \cE_{\ell_a-1}$ and $h \in \cE_{\ell_a-1}$ and then calculate a rank-maximal matching in the instance without the already matched agents and houses. Afterwards, we check if this matching is indeed NRM and return it if it is, otherwise we return that no NRM matching exists. 

If there is no NRM matching this is indeed correct. Otherwise, assume that an NRM matching $M'$ exists and let $M$ be the matching returned by our algorithm. By our previous proof, we know that the agents matched to unrevealed preferences by $M'$ and $M$ are the same. Further, for any other preference extension, the signature over the agents matched to revealed preferences is obviously the same, since $M$ is rank-maximal on these agents. Thus $M'$ and $M$ have the same signature for all possible preference extensions and therefore $M$ is NRM.
\end{proof}
Before turning to the correctness proof of \Cref{alg:rm_hq} we show a few auxiliary lemmas for the hybrid query model. For a given preference profile $\succ$ and an NRM matching $M$ with partial preferences $\rank'$ of $\succ$, we call an agent $a \in A$ a prefix agent, if the ranks $1$ to $r_a-1$ are all revealed for $\rank$. Similarly, we call all other agents non-prefix agents.  We can show that all houses not ranked by a non-prefix agent $a$ can be guaranteed to be matched to an agent that rank them at most at rank $k_a$, \ie the first rank for which agent $a$ has not revealed a house yet.
\begin{restatable}{lem}{hybridunm}
Let $a \in A$ be a non prefix agent. Then for any house $h \notin \rev(a)$ with $h \neq M(a)$ there exists a $k \le k_a$ with $h \notin \cE_k$.
\label{lem:rm_hybrid_unm}
\end{restatable}
\begin{proof}
Assume that there is a $h \notin \rev(a)$ with $h \in \cE_k$ for every $k \le k_a$ and assume that the preferences are completed such that $\rank(a, h) = k_a$, \ie the smallest unrevealed rank, with all other preferences consistent with $\succ$. Since $a,h \in \cE_i$ for any $i < k_a$ there must exist an edge between $a$ and $h$ in $G_{k_a}$ . However, since $h \in \cE_{k_a}$, when only looking at $\succ$, when $h$ is not matched to $a$, there is an alternating path from $h$ to an unmatched house $h'$ not containing $a$. Hence, we could augment this matching along this path and by matching $a$ to $h$. Thus, $a$ had to be matched to $h$ in the first place. 
\end{proof}
Similarly, we can show that a non-prefix agent must reveal a house for each rank matched by a rank-maximum matching which is better than his own rank.
\begin{restatable}{lem}{hybridrankmatched}
Let $a \in A$ be a non-prefix agent. Then for all $a' \in A$ with $\rank(a', M(a')) < r_a$ it has to either hold that $M(a') \in \rev(a)$ or for all $r' < \rank(a', M(a'))$ there has to be a $h' \in H$ with $h' \in \rev(a)$ and $\rank(a, h') = r'$.
\label{lem:hybrid_rank_matched}
\end{restatable}
\begin{proof}
Assume that $\rank(a', M(a')) < r_a$ and $M(a') \notin \lvert \rev(a) \rvert$. Since $M$ matches $a'$ to $M(a')$ we know that $M(a') \in \cE_{r'}$ for all $r' < \rank(a', M(a')) $. However, since also $a \in \cE_{r'}$ there cannot be an edge between $a$ and $M(a')$ in $G_{r'}$ in any preference completion and thus $a$ must elicit a preference with rank $r'$. Otherwise, $a$ could elicit $M(a')$ with rank $r'$.
\end{proof}

\algormhq*
\begin{proof}
It is easy to see that the matching computed by \Cref{alg:rm_hq} is indeed necessarily rank-maximal, since any agent $a \notin U$ in line 26 will be in either $\cU$ or $\cO$ and will thus be matched to a house they have already revealed. Further, the only houses unrevealed by an agent $a \in U$ in line 26 are houses which are in $\cU$ or $\cO$ and can thus not be matched to $a$ in any rank-maximal matching by \Cref{lem:rm_char}. Thus, the preferences elicited in \cref{alg:rm_hq} are a superset of the edges of $G_n$ which by \Cref{lem:rm_char} implies that we correctly find a necessarily rank-maximal matching.

Let OPT be the optimal algorithm and OPT$_P$ be the set of prefix agents after the partial preferences were revealed by OPT and OPT$_{NP}$ the non-prefix agents of OPT. Similarly, let ALG$_{P}$ and ALG$_{NP}$ be the prefix and non-prefix agents of our algorithm. For simplicity, we also let $i$ and $c$ be the respective values of the variables at the end of our algorithm. For any agent $a \in A$ let $\opt_a$ be the number of queries asked to $a$ by an optimal algorithm and $\alg_a$ the number of queries asked by our algorithm. Our goal is now to bound $\frac{\alg_a}{\opt_a} \le 6$ for any agent, thus showing that the competitive ratio is at most $6$. Before doing this, we must however discuss the case, where $\rev(a) = \emptyset$ for some agent $a \in A$. Firstly, assume that there is some $a' \in A$ with $\rank(a', M(a') > 1$. Then we can extend the preferences in such a way that $\rank(a, M(a') = 1$ and the matching is not NRM. Thus, it needs to hold that all agents except for $a$ are matched to their first choice. Further, similar to the proof in \Cref{thrm:32-comp} we must know that $\rank(a',M(a)) = n$ for all agents. Thus, any optimal algorithm must ask at least $2(n-1)$ queries. Further, our algorithm asks every agent one question and then queries $a$ and the agent sharing the first choice agent with $a$ both exactly two queries. Since we can assume that $n \ge 3$ we get an upper bound on the competitive ratio of $\frac{n + 4}{2(n-1)} \le \frac{7}{4} \le 2$.

Now we can assume that $\rev(a) \neq \emptyset$ and distinguish four cases:
\begin{itemize}
    \item For any $a \in A$ with $a \in$  OPT$_P$  and $a \in$  ALG$_P$ we know that our algorithm asks at most $r_a$ queries to $a$ while the optimal algorithm asks at least $r_a-1$. Thus, the quotient for $a$ is upper bounded by \[\frac{\alg_a}{\opt_a} \le \frac{r_a}{r_a-1} \le \frac{2}{1} = 2.\]
    \item If $a \in$ OPT$_P$  and $a \in$  ALG$_{NP}$ we know that $r_a > i$ and thus $r_a -1 \ge i \ge c \ge \lvert V \rvert$. Thus, the quotient for $a$ in this case is upper bounded by \[\frac{\alg_a}{\opt_a} \le \frac{i + \lvert V \rvert}{i} \le 1 + \frac{\lvert V \rvert}{\lvert V \rvert} = 2.\] 
    \item Next, assume that $a \in \opt_{NP}$ and $a \in \alg_P$. Then we know that in iteration $i = r_a$ it has to hold that $c \le \lvert V \rvert$. By \Cref{lem:rm_hybrid_unm} we know that all but one house not in $\cU_{r_a}$ and $\cO_{r_a}$, \ie in $V$, need to be queried by $a$ in the optimal algorithm. Similarly, by \Cref{lem:hybrid_rank_matched} we know that at most $c + 1$ ranks before $r_a$ can be unqueried, namely all ranks $r'$ without any agent $a'$ such that $\rank(a', M(a'))$ as well as the largest $r'$ such that such an $a'$ exists. Since no house with a rank smaller than $r_a$ can belong to $\cE_{r_a}$ we get $\opt_a \ge r_a - c + \lvert V \rvert -2 \ge r_a - 2  + c - c$. Thus, we get an upper bound on the quotient of \[ \frac{\alg_a}{\opt_a} \le \frac{r_a}{r_a - 2 } \le  3\] since we know that at least one query is asked to $a$.
    \item Finally, the case of $a \in \opt_{NP}$ and $a \in \alg_{NP}$ remains. Let $p_a$ be the length of the prefix of the preferences queried by the optimal algorithm.
    
    If $p_a > i$, we get that $a$ gets asked at least $i \ge \lvert V \rvert$ queries by the optimal algorithm. Since our algorithm asks at most $i + \lvert V \rvert $ queries, this implies a quotient of $2$. 
    
    If $p_a \le i$, then we know by \Cref{lem:rm_hybrid_unm} that the optimal algorithm must ask at least $\lvert V \rvert + p_a-1$ queries to $a$. If $p_a = i$, this implies a quotient of \[\frac{\alg_a}{\opt_a} \le \frac{i + \lvert V \rvert}{\lvert V \rvert + p_a-1} = \frac{i + \lvert V \rvert}{i + \lvert V \rvert -1}\le 2.\] If $p_a < i$, we know that every but one house in $V$ must be queried by the optimal algorithm and at most $c_{i-1} + 1 \le \lvert V \rvert + 1$ ranks can be left unqueried by the optimal algorithm, thus leading to a quotient of at most \[\frac{\alg_a}{\opt_a} \le \frac{i + \lvert V \rvert }{(\lvert V \rvert -1) + (i - \lvert V \rvert -1) } = \frac{i + \lvert V \rvert}{i-2} \le  6\] due to $\lvert V \rvert \le c \le i$ and since we can assume that at least one query has been asked to $a$.
\end{itemize}

Thus, the competitive ratio is at most $6$. 
\end{proof}
\end{document}